\documentclass[10point]{article}       
\usepackage{mathptmx}      
\usepackage{a4wide}
\usepackage[latin1]{inputenc}

\usepackage{amsthm,amsmath,amssymb}
\usepackage{url}
\newtheorem{theorem}{Theorem}
\newtheorem{proposition}{Proposition}
\newtheorem{remark}{Remark}

\DeclareMathOperator{\KS}{\mathit{C}\,}
\DeclareMathOperator{\KP}{\mathit{K}\,}

\begin{document}

\title{An additivity theorem for plain Kolmogorov complexity
\thanks{B. Bauwens is
supported by Fundaçao para a Ciência e a Tecnologia by grant 
(SFRH/BPD/75129/2010), and is also partially supported by 
project $CSI^2$ (PTDC/EIAC/099951/2008).
A. Shen is supported in part by projects NAFIT ANR-08-EMER-008-01 grant and RFBR 09-01-00709a. Author are grateful to their colleagues for interesting discussions and to the anonymous referees for useful comments.}} 


\author{Bruno Bauwens\footnote{ 
  Instituto de Telecomunicações Faculdade de Ci\^encia at Porto University, 
  R.Campo Alegre, 1021/1055, 4169-007 Porto - Portugal.
  }
        \and Alexander Shen\footnote{
  Laboratoire d'Informatique, de Robotique et de Microlectronique de Montpellier, 
  UMR 5506 - CC 477, 161 rue Ada,  34095 Montpellier Cedex 5 France. 
  }
}



\maketitle

\begin{abstract}
We prove the formula $\KS(a,b)=\KP(a|\KS(a,b))+\KS(b|a,\KS(a,b))+O(1)$ that expresses the plain
complexity of a pair in terms of prefix-free and plain conditional complexities of its components.
\end{abstract}

\vspace{1cm}

The well known formula from Shannon information theory states that $H(\xi,\eta)=H(\xi)+H(\eta|\xi)$.
Here $\xi,\eta$ are random variables and $H$ stands for the Shannon entropy. A similar formula for
algorithmic information theory was proven by Kolmogorov and Levin~\cite{kolm68} and says that
\begin{equation*}
  \KS(a,b)=\KS(a)+\KS(b|a)+O(\log n),
\end{equation*}
where $a$ and $b$ are binary strings of length at most $n$
and $\KS$ stands for Kolmogorov complexity (as defined initially by Kolmogorov~\cite{kolm65}; now
this version is usually called \emph{plain} Kolmogorov complexity). Informally, $\KS(u)$ is the
minimal length of a program that produces $u$, and $\KS(u|v)$ is the minimal length of a program
that transforms $v$ to $u$; the complexity $\KS(u,v)$ of a pair $(u,v)$ is defined as the complexity
of some standard encoding of this pair. 

This formula implies that $I(a:b)=I(b:a)+O(\log n)$ where $I(u:v)$ is the amount of information in
$u$ about $v$ defined as $\KS(v)-\KS(v|u)$; this property is often called ``symmetry of
information''. The term $O(\log n)$, as was noted in~\cite{kolm68}, cannot be replaced by $O(1)$.
Later Levin found an $O(1)$-exact version of this formula that uses the so-called \emph{prefix-free}
version of complexity: 
\begin{equation*}
\KP(a,b)=\KP(a)+\KP(b|a,\KP(a))+O(1);
\end{equation*}
this version, reported in \cite{Gacs73}, was also discovered by Chaitin~\cite{chaitin75}. 
In the definition of prefix-free complexity we restrict ourselves to self-delimiting programs:
reading a program from left to right, the interpreter determines where it ends. See, e.g.,
\cite{shen00} for the definitions and proofs of these results.

In this note we provide a somewhat similar formula for plain complexity (also with $O(1)$-precision):

\begin{theorem}\label{main}
  \begin{equation*}
   \KS(a,b) = \KP(a|\KS(a,b)) + \KS(b|a, \KS(a,b))  + O(1) \,.
  \end{equation*}
  \label{th:additivityCK}
\end{theorem}

\begin{proof}
   The proof is not difficult after the formula is presented. The $\le$-inequality is a
   generalization of the inequality $\KS(x,y)\le\KP(x)+\KS(y)$ and can be proven in the same way.
   Assume that $p$ is a self-delimiting program that maps $\KS(a,b)$ to $a$, and $q$ is a (not
   necessarily self-delimiting) program that maps $a$ and $\KS(a,b)$ to $b$. The natural idea is to
   concatenate $p$ and $q$; since $p$ is self-delimiting, given $pq$ one may find where $p$ ends and $q$ starts, and then use $p$ to get $a$ and $q$ to get $b$.  However, this idea needs some refinement: in both cases we need to know $\KS(a,b)$ in advance; one may use the length of $pq$ as a replacement for it, but since we have not yet proven the equality, we have no right to do so.
   
So more caution is needed. Assume that the $\le$-inequality is not true and $\KS(a,b)$ exceeds $\KP(a|\KS(a,b)) + \KS(b|a, \KS(a,b))$ by some $d$. Then we can concatenate prefix-free description $\bar {d}$ of $d$ (that has length $O(\log d)$), then $p$ and then $q$. Now we have enough information: first we find $d$, then $\KS(a,b)=|p|+|q|+d$, then $a$, and finally $b$. Therefore $\KS(a,b)$ does not exceed $O(\log d)+|p|+|q|+O(1)$, therefore $d\le O(\log d)+O(1)$ and $d=O(1)$. The $\le$-inequality is proven.

Let us prove the reverse inequality. In this proof we use the interpretation of prefix-free
complexity as the logarithm of a priori probability (see, e.g., \cite{shen00} for details). If
$n=\KS(a,b)$ is given, one can start enumerating all pairs $(x,y)$ such that $\KS(x,y)\le n$; there
are at most $2^{n+1}$ of them and the pair $(a,b)$ is among them. For fixed $x$, for each pair $(x,y)$ in this
enumeration we add $2^{-n-1}$ to the probability of $x$; in this way we approximate (from below) the
semimeasure $P(x|n)=N_x2^{-n-1}$. Therefore, we get an upper bound for $\KP(a|n)$:
\begin{equation*}
    \KP(a|n)\le -\log P(a|n) + O(1) = n - \log_2 N_a+O(1) \,,
  \end{equation*}
where $N_a$ is the number of $y$'s such that $\KS(a,y)\le n$. On the other hand, given $a$ and $n$, we
can enumerate all these $y$, and $b$ is among them, so $b$ can be described by its ordinal number in
this enumeration, therefore
\begin{equation*}
   \KS(b|a,n)\le \log_2 N_a+O(1) \,.
 \end{equation*}
Summing these two inequalities, we get the desired result.       
\end{proof}

We can now get several known $O(1)$-equalities for complexities as corollaries of Theorem~\ref{main}.

\begin{itemize}
\item Recall that $\KS(a,\KS(a))=\KS(a)$, and $\KP(a,\KP(a))=\KP(a)$ (the $O(1)$-additive terms are omitted here and below), since the shortest program for $a$ also describes its own length.
\item For empty $b$ we get $\KS(a)=\KP(a|\KS(a))$, see also \cite{GacsNotes,LiVitanyi}. 
\item For empty $a$ we get $\KS(b)=\KS(b|\KS(b))$, see also \cite{GacsNotes,LiVitanyi}.
\item The last two equalities imply that $\KS(u|\KS(u))=\KP(u|\KS(u))$. 

The direct proof for last three statements is also easy. To show that $\KS(a)\le \KS(a|\KS(a))$,
assume that some program $p$ maps $\KS(a)$ to $a$ and is $d$ bits shorter than $\KS(a)$. Then we add
a prefix $\bar d$ of length $O(\log d)$ that describes $d$ in a self-delimiting way, and note that
$\bar d p$ determines first $\KS(a)$ and then $a$, so $d\le O(\log d)+O(1)$ and $d=O(1)$. To show
that $\KP(a|\KS(a))\le \KS(a|\KS(a))$ we note that in the presence of $\KS(a)$ every program of
length $\KS(a)$ can be considered as a self-delimiting one, since its length is known. 

Levin also pointed out that $\KS(a)$ can be defined in terms of prefix-free complexity as a minimal
$i$ such that $\KP(a|i)\le i$. (Indeed, for $i=\KS(a)$ both sides differ by $O(1)$, and changing
right hand side by $d$, we change left hand side by $O(\log d)$, so the intersection point is unique
up to $O(1)$-precision. In other terms, $\KP(a|i) = i+O(1)$ implies $\KS(a) = i+O(1)$.)
   
\item More generally, we may let $a$ be some fixed computable function of $b$: if $a=f(b)$, we get
  $\KS(b)=\KP(f(b)|\KS(b))+\KS(b|f(b),\KS(b))$.  

\end{itemize}

One can also see that Theorem~\ref{main} can be formally derived from Levin's results mentioned above. To show that 
\begin{equation*}
\KS(b|a,\KS(a,b))= \KS(a,b)-\KP(a|\KS(a,b))    
\end{equation*}
we need to show that the right hand side $i=\KS(a,b)-\KP(a|\KS(a,b))$ satisfies the equality
$\KP(b|a,\KS(a,b),i)= i$ with $O(1)$-precision, which implies $\KS(b|a,\KS(a,b)) = i$.
(We omit all $O(1)$-terms, as usual.) In the
condition of the last inequality we may replace $i$ by $\KP(a|\KS(a,b))$ since $\KS(a,b)$ is already
in the condition. Therefore, we need to show that
\begin{equation*}
\KP(b|a,\KS(a,b),\KP(a|\KS(a,b)))= \KS(a,b)-\KP(a|\KS(a,b))
\end{equation*}
or
\begin{equation*}
\KP(b|a,\KS(a,b),\KP(a|\KS(a,b)))+\KP(a|\KS(a,b))= \KS(a,b).
\end{equation*}
But the sum in the left hand side equals $\KP(a,b|\KS(a,b))$ due to the formula for prefix
complexity of a pair $(a,b)$ relativized to the condition $\KS(a,b)$, and it remains to note that
$\KP(a,b|\KS(a,b))=\KS(a,b)$. (This alternative proof was suggested by Peter Gacs.)
   
We can obtain a different version of Theorem~\ref{main}:

\begin{proposition}
\begin{equation*}
 \KS(a,b) = \KP(a|\KS(a,b)) + \KS(b|a, \KP(a|\KS(a,b)))  + O(1) \,.     
 \end{equation*}
 \end{proposition}
 \begin{proof} Indeed, the $\leq$-inequality can be shown in the same way as the $\leq$-inequality  in the proof of Theorem~\ref{th:additivityCK}, hence it remains to show the $\geq$-inequality.
Let $p$ be a program of length $\KS(b|a,\KS(a,b))$ that computes $b$ given $a$ and
$\KS(a,b)$. (The program $p$ is not assumed to be self-delimiting.) Knowing $p$, we can also compute
$b$ given $a$ and $\KP(a|\KS(a,b))$. First, we compute $|p| + \KP(a|\KS(a,b))$, and this sum equals
$C(a,b)$ (Theorem~\ref{main}). Then, using $a$ again, we compute $b$. Hence $\KS(b|a, \KS(a,b)) \geq
\KS(b|a, \KP(a|\KS(a,b)))$.
\end{proof}

One may complain that Theorem \ref{th:additivityCK} is a bit strange since it uses prefix-free
complexity in one term and plain complexity in the second (conditional) part. As we have already
noted, one cannot use $\KS$ in both parts: $\KS(a,b)$ can exceed even $\KS(a)+\KS(b)$ by a
logarithmic term. One may then ask whether it is possible to exchange plain and prefix-free complexity in the two terms we have and prove that $\KS(a,b)$ equals something like
\begin{equation*}
\KS(a|\KS(a,b)) + \KP(b|a,\KS(a,b)).
\end{equation*}
It turns out that it is not possible: even the inequality $\KS(a,b)\le\KS(a)+\KP(b|a)+O(1)$ is not true. At first it seems that one could concatenate a self-delimiting program $q$ that produces $b$ given $a$ and a (plain) program $p$ that produces $a$, in the hope that the endpoint of $q$ can be reconstructed, and then the rest is $p$. However, this idea does not work: the program $q$ is self-delimiting only when $a$ is known; to know $a$ we need to have $p$, and to know $p$ we need to know where $q$ ends, so there is a vicious circle here.

Let us show that the problem is unavoidable and that for infinitely many pairs $(x,y)$ we have 
\begin{equation*}
\KS(x,y) \geq  \KS(x) + \KP(y|x) + \log n - 2 \log \log n-O(1),
\end{equation*}
where $n=|x|+|y|$ is the total length of both strings. To construct such a pair, let $n = 2^k$ for some $k$, and choose a string $r$ of length $n$ and a natural number $i<n$ such that $\KS(r,i) \geq n + \log n$. (For every $n$, there are $n2^n$ pairs $(r,i)$, so one of them has high complexity.)

Let $x = r_1\dots r_i$ and $y = r_{i+1}\dots r_n$. Note that $\KS(x,y) = \KS(r,i) \geq n + \log n$
and that $\KS(x) \leq i$. Furthermore, $\KP(y|x) \le \KP(y|x,n)+\KP(n)$. Here $\KP(y|x,n)\le
|y|=n-i$, since $x$ and $n$ determine $|y|$ and $\KP(y\mid |y|)\le |y|$; on the other hand,
$\KP(n)\le 2\log\log n$.\footnote{As a byproduct of this example and the discussion above we conclude that $\KP(x|y)$ cannot be defined as minimal prefix-free complexity of a program that maps $y$ to $x$: the value $\KP(y|x)$ can be smaller than $\min \left\{ K(p): U(p,x) = y \right\}$, where $U$ is the universal function. Indeed, in this case we would have the inequality $\KS(x,y)\le\KS(x)+\KP(y|x)$, since the prefix-free description of a program that maps $x$ to $y$ and a shortest description for $x$ can be concatenated into a description of the pair $(x,y)$.}

\medskip

There is still some chance to get a formula for the plain complexity of a pair $(x,y)$ that involves only plain complexities, assuming that we add some condition in the left hand side, i.e., to get some formula of the type $\KS(a,b|?)=?$. Unfortunately, the best result in this direction that we managed to get is the following observation:

     \begin{proposition}
For all $x,y$ there exists a \textup(unique up to $O(1)$-precision\textup) pair $(k,l)$ such that $\KS(x|l) = k$, $\KS(y|x,k) =l$. For such a pair we have $\KS(x|l)=k$, $\KS(y|x,k) = l$ and  this implies $\KS(x,y|k,l)=\KS(x,y|k) = \KS(x,y|l) = l+k$ \textup(all with $O(1)$-precision\textup).
     \end{proposition}

\begin{proof}
The pair in question is a fixed point of $F\colon (k,l)\mapsto (\KS(x|l),\KS(y|x,k))$. It exists and is unique since $F$ maps points at distance $d$ into points at distance $O(\log d)$. (Here ``distance'' means geometric distance between points in $\mathbb{Z}^2$.)

Using the relativized version of the statement $\KS(z)=\KS(z|\KS(z))$, we conclude that
$\KS(x|k,l)=k$ and $\KS(y|x,k,l)=l$. Let us prove now that $\KS(x,y|k,l)=k+l$. Indeed, the standard
proof of Kolmogorov--Levin theorem shows that for any $x,y,k',l'$ such that 
\begin{equation*}
  \KS(x,y|k',l') \leq k' + l' 
\end{equation*} 
we have either   
\begin{equation*} 
\KS(x|k',l') \leq k'
    \;\;\; \text{ or } \;\;\; \KS(y|x,k',l') \leq l' \,.  \end{equation*} Hence if $\KS(x|k,l) = k$
    and $\KS(y|x,k,l) = l$ for some $k$ and $l$, we have $\KS(x,y|k,l) \geq k+l$ (otherwise $k$ and
    $l$ can be decreased to get a contradiction). By concatenation we obtain also that $\KS(x,y|k,l)
    \le  k+l$, so $\KS(x,y|k,l)=k+l$ (all equations with $O(1)$-precision).
 
It remains to show that $\KS(x,y|k,l)=k+l$ implies $\KS(x,y|k)=k+l$ and, similarly,
$\KS(x,y|l)=k+l$. Indeed, a program of length $k+l$ that maps $(k,l)$ to $(x,y)$, can be used to map
$k$ (or $l$) to $(x,y)$: knowing the length of the program and one of the values of $k$ and $l$, we
reconstruct the other value.  \end{proof}

\begin{remark}
One can ask what can be said about pairs $(k',l')$ such that $\KS(x|l')\le k'$ and $\KS(y|x,k')\le l'$. 
The pair $(k,l)$ given by the theorem is not necessarily coordinate-wise minimal: for example,
taking a large $k'$ that contains full information about $y$ we may let $l'=0$. Indeed, $\KS(x|0)\le
k'$ (since $k'$ is large) and $\KS(y|x,k')\le 0$ (since $k'$ determines $y$). However, to get some
decrease in $k'$ (compared to $k$) or $l'$ (compared to $l$) we need to change the other parameter
by an exponentially bigger quantity, since the information distance between $i$ and $i'$ is
$O(\log|i-i'|)$. The change in the other parameter should be its increase, otherwise we could repeat
the arguments exchanging $k$ and $l$ and get a contradiction (each of two changes could not be
exponentially big compared to the other one). 

\end{remark}


\end{document}